\documentclass[12pt]{article}

\usepackage{graphicx, amssymb, latexsym, amsfonts, amsmath, lscape, amscd,
amsthm, color, epsfig, mathrsfs, tikz, enumerate}
\usepackage[normalem]{ulem} 

\usepackage{enumitem}
\newlist{steps}{enumerate}{1}
\setlist[steps, 1]{leftmargin=1.5cm, label = Step \arabic*}

\newtheorem{theorem}{Theorem}[section]

\newtheorem{corollary}[theorem]{Corollary}

\newtheorem{lemma}[theorem]{Lemma}

\newtheorem{property}[theorem]{Property}

\newtheorem{algorithm}[theorem]{Algorithm}
\newtheorem{observation}[theorem]{Observation}

\newtheorem{claim}{Claim}

\newcommand\DELETE[1]{}

\begin{document}


\title{{\bf Maximum Directed Linear Arrangement}}
\author{
{\sc Matt DeVos}$\,^{a}$, {\sc Kathryn Nurse}$\,^{b}$\\
\mbox{}\\
{\small $(a)$ Simon Fraser University, Burnaby Canada}\\
{\small $(b)$ \'Ecole Normale Sup\'erieure, Paris France}
}

\date{\today}

\maketitle

\begin{abstract}
We present a new problem on digraphs, Maximum Directed Linear Arrangement (MaxDLA). This is a directed and maximization variant of the well-studied Minimum Linear Arrangement (MinLA) problem. We relate MaxDLA to the Maximum Directed Cut (MaxDiCut) problem by bounding each in terms of the other. We prove that both MaxDiCut and MaxDLA are NP-Hard for planar digraphs. By contrast, the undirected Maximum Cut problem is known to be polynomial-time solvable on planar graphs. We present a polynomial-time algorithm for MaxDLA on orientations of bounded-degree trees, and, as a by-product, a polynomial-time algorithm for MinLA on graphs $G$ when $\overline{G}$ is a bounded-degree tree. This complements the known fact that MinLA is polynomial-time solvable on trees. Finally, motivated by Harper's celebrated isoperimetric inequality for hypercubes, we prove that tournaments, transitive acyclic digraphs, and orientations of graphs with maximum degree at most two have an arrangement so that every cut is maximum for its position in the arrangement.
\end{abstract}

\section{Introduction}

Except where otherwise noted, we shall assume that all directed graphs (digraphs) are {simple} (i.e. do not have loops or parallel edges) but may have opposite edges (an edge $(u,v)$ is opposite $(v,u)$). We refer the reader to \cite{Diestel3rdEdition} for terms not defined here. A digraph is \emph{symmetric} if every edge has an opposite, and in this case we consider it a (undirected) graph.

Let $D$ be an $n$-vertex digraph. A (\emph{linear}) \emph{arrangement} $\pi$ of $D$ is a bijection from $V(D)$ to $\{1,\ldots, n\}$. We will write $\pi = (v_1, v_2, \ldots, v_n)$ to indicate that $\pi(v_i) = i$ for every $1 \le i \le n$. Given an arrangement $\pi$ of a digraph $D$, the \emph{value of an edge} $e=(u,v)$ of $D$ is defined to be 
$val_{\pi}(e) = \max\{ 0, \pi(v) - \pi(u) \}$.  The \emph{value of $\pi$} is defined as $val_D(\pi) = \sum_{e \in E(D)} val_{\pi}(e)$.  The Maximum Directed Linear Arrangement Problem can now be stated as follows.

\medskip

\textsc{Maximum Directed Linear Arrangement (MaxDLA)}

\textbf{Input:} $n$-vertex digraph $D$, integer $k$

\textbf{Problem:} Is there an arrangement $\pi$ of $D$ with value at least $k$?

\medskip

We denote by $maxDLA(D)$ the largest such $k$.
MaxDLA is related to two well-studied problems: Minimum Linear Arrangement, and Simple Maximum Directed Cut. 

The Minimum Linear Arrangement problem (MinLA) asks, given an undirected graph $G$ and integer $k$, if there is an arrangement of $G$ with value at most $k$. MinLA is NP-complete in general \cite[GT42]{garey76}, even for bipartite \cite{shiloach79}, and interval \cite{cohen06} graphs, but is polynomial-time solvable for many classes of graphs such as trees \cite{goldklip, shiloach79, chung84}, and unit interval graphs \cite{Jinjiang1995,safro02}. See \cite{diaz02,petit11} for a survey. 

The relation between MinLA and MaxDLA is straightforward. Given a digraph $D$, its \emph{complement} $\overline{D}$ is the digraph with vertex set $V(D)$ and edge set $\overline{E(D)}$. Observe that, for an arrangement $\pi$ of an $n$-vertex digraph $D$, $val_D(\pi)$ is maximum exactly when $val_{\overline{D}}(\pi)$ is minimum. Indeed,  $val_D(\pi) + val_{\overline{D}}(\pi) = \binom{n+1}{3}$, which is the value of an arrangement of an $n$-vertex complete digraph \cite[p. 244]{garey76}. Theorem 1.1 follows.

\begin{theorem}
MaxDLA is NP-Complete.
\end{theorem}

As stated previously, MinLA is solvable in polynomial time for trees. This was first shown by Goldberg and Klipker in 1976 with an $O(n^3)$ algorithm \cite{goldklip}. And improved by Shiloach in 1979, with an interesting $O(n^{2.2})$ algorithm \cite{shiloach79}. Finally, in 1984, Chung improved Shiloach's algorithm and, with a careful analysis of running time, gave a $O(n^\lambda)$ algorithm where $\lambda \approx 1.6$ \cite{chung84} which is the current best. In Section \ref{ortrees} of this paper, we present a polynomial-time algorithm solving MinLA for the complements of trees with degree bounded by a constant. This was discovered first in the context of MaxDLA. The algorithm for MaxDLA, which is polynomial-time for orientations of trees with constant-bounded degree, occupies the bulk of the section. Hence, the main results of section 4 are the following Theorem and Corollary.

\begin{theorem}
Let $G$ be a forest on $n$ vertices with $\Delta (G) = d$ and let $D$ be an orientation of $G$. Then the MaxDLA of $D$ is solvable in time $O(n^{4d})$.
\end{theorem}

\begin{corollary}
Let $G$ be a graph on $n$ vertices whose complement is a forest of maximum degree $d$. Then MinLA of $G$ is solvable in time $O(n^{4d})$.
\end{corollary}

The Simple Maximum Directed Cut problem (MaxDiCut) is also related to MaxDLA. Given a partition of the vertices of a digraph $D$ into two sets $S$ and $T$, the \emph{directed cut from $S$ to $T$}, written $E(S,T)$, is the set of edges with tail in $S$ and head in $T$. MaxDiCut asks, given a digraph $D$ and integer $k$, if there is a directed cut of $D$ containing at least $k$ edges. MaxDiCut is NP-complete in general \cite[pp. 244-246]{garey76}, and remains so for many symmetric graph classes. That is, the undirected version is NP-complete for chordal, tripartite, split \cite{Bodlaender91}, and unit-disk graphs \cite{Diaz2007}, but is solvable in polynomial time for cographs \cite{Bodlaender91} and planar graphs \cite{Hadlock1975}. The following theorem, proved at the beginning of Section 3, contrasts the last result. 

\begin{theorem}
MaxDiCut is NP-complete even when restricted to planar multi-digraphs of maximum degree 14.  
\end{theorem}

As a corollary we show that the same holds for MaxDLA.

\begin{corollary}
MaxDLA is NP-complete even when restricted to planar multi-digraphs of maximum degree 14.  \end{corollary}

In Section 3 we also show a dependence between MaxDLA and MaxDiCut, whose bounds are best possible.

\begin{theorem}
For every digraph $D$, 
$$ \tfrac{n}{2} MaxDiCut(D) \leq MaxDLA(D) \leq (n-1)MaxDiCut(D). $$
\end{theorem}

For any arrangement $(v_1, v_2, \ldots, v_n)$ of a directed graph $D$, there are $n-1$ associated directed cuts of the form $E( \{v_1, \ldots, v_k \}, \{v_{k+1}, \ldots, v_n \} )$ where $1 \le k \le n-1$.  We call these the \emph{cuts} of the arrangement.  

The last section of this paper examines digraphs for which the maximum directed linear arrangement has the property that every cut of the arrangement $E( \{v_1, \ldots, v_k \},\allowbreak \{v_{k+1}, \ldots, v_n \} )$ is a largest directed cut in $G$ of the form $E(X,Y)$ where $|X|=k$.   This is inspired by a famous edge isoperimetric inequality of Harper.  In 1966, Harper gave an arrangement of the vertices of a hypercube so that every cut of the arrangement was minimum over all cuts separating the same size vertex sets \cite{Harper1966}. In contrast, in Section \ref{sect:wonderful} we present three classes of digraphs, all with an arrangement where every cut is maximum over all cuts separating the same size vertex sets.  These digraphs are tournaments, orientations of graphs $G$ with $\Delta(G) \leq 2$, and transitive acyclic digraphs.

\section{Properties of Directed Linear Arrangements}\label{sect:properties}

In this section we establish some basic properties of Directed Linear Arrangements that will be helpful in our investigations.  

Fix an arrangement $\pi = (v_1, v_2, ..., v_n)$ of a digraph $D$ and define $S_i = \{v_j \in V(D) \mid j \leq i \}$ and $T_i = \{ v_j \in V(D) \mid j > i \}$.  Then the cuts $C_1, \ldots, C_{n-1}$ of the arrangement $\pi$ are defined by the rule that $C_i$ is the directed cut from $S_i$ to $T_i$.  We say that the arrangement $\pi$ \emph{contains} each $C_i$ and we let $c_i = |C_i|$. When needed, we write $c_i(\pi)$ to specify the arrangement $\pi$.
%
%

Observe that the value of an arrangement $\pi$ can be calculated by summing its cuts 
\begin{equation*}
val(\pi) = \sum_{i=1}^{n-1}c_i.
\end{equation*}

The \emph{level} of a vertex $v_i$ in an arrangement $\pi = (v_1, v_2, ..., v_n)$ can be thought of as its contribution to $c_i$, and is defined as follows 

\begin{equation*}
l_\pi(v_i) = c_{i} - c_{i-1}\text{,}
\end{equation*}
where we let $c_0$ = $c_n$ = 0.

The value of a cut $c_i$ is therefore the sum of the levels to its left. And so there are three equivalent ways to calculate the value of an arrangement, based on edges, cuts, and levels respectively.

\begin{property}\label{prop:calculate value using edges cuts levels}
The value of an arrangement $\pi = (v_1, v_2, ..., v_n)$ of a digraph $D=(V,E)$ is \begin{equation*}
val(\pi) = \sum_{e \in E(D)}val(e) = \sum_{i=1}^{n-1}c_i = \sum_{i=1}^{n-1}\sum_{j=1}^i l
_\pi(v_j).
\end{equation*}
\end{property}

Levels lead to a nice abstraction of MaxDLA. Let $\pi = (v_1, v_2, ..., v_n)$ be an arrangement of digraph $D$. Given vertex $v_i$, its level in $\pi$ is the number of its in-neighbours to the its left subtracted from the number of its out-neighbours to its right $l_\pi(v_i) = |N^+(v_i) \cap T_i| - |N^-(v_i) \cap S_i|$. This is just a direct application of the definition of directed cut. From this, by adding and subtracting $|N^+(v_i) \cap S_i|$, it follows that the level of a vertex is 
\begin{equation}\label{eqn: level calculation}
l_\pi(v_i) = d^+(v_i) - |N(v_i) \cap S_i|.
\end{equation}
Therefore the levels, and hence the value, of an arrangement $\pi$ can be calculated without knowing the direction of the edges of $D$, only the out-degree of each vertex. Hence we have the following property. 


\begin{property}
Let $G$ be a graph, and let $D$ and $D'$ be orientations of $G$ so that every vertex has the same outdegree in both $D$ and $D'$. If $\pi$ is an arrangement of $D$, then $val_D(\pi) = val_{D'}(\pi)$.
\end{property}


We say the \emph{signature} $s$ of an arrangement $\pi = (v_1, v_2, ..., v_n)$ is the $(n-1)$-tuple of the sizes of its cuts
\begin{equation}\label{eqn: signature of cuts}
s(\pi) = (c_1, c_2, ..., c_{n-1}).
\end{equation}
And we use the shorthand $s_i$ to mean the $i$\textsuperscript{th} element of $s$. Then $val\big(s(\pi)\big) = \sum_{i=1}^{n-1}c_i =  val(\pi)$.

For a digraph $D$, let $S$ be the set of signatures of $D$. That is, $s \in S$ whenever there is an arrangement $\pi$ of $D$ so that $s(\pi) = s$. For $s, s' \in S$ we write $s' \leq s$ whenever $s'_i \leq s_i$ for all $i$. This notation captures the idea of $s$ being no worse than $s'$ on every cut. It is easy to see that $S$ is a partial order under $\leq$. We say an arrangement $\pi$ is maximum (maximal) if $s(\pi)$ is maximum (maximal) in $S$.

\begin{observation}\label{obs: disconnected}
Let $D$ be a disconnected digraph and let $H$ be one of its components. An arrangement $\pi$ is a maximal (maximum) arrangement of $D$ only if $\pi$ restricted to $H$ is a maximal (maximum) arrangement of $H$.
\end{observation}

This motivates the following slight generalization of our problem. 
The Maximum Directed Linear Arrangement (MaxDLA) problem is contained in a problem on undirected graphs with vertex weights. We define the Weighted Maximum Linear Arrangement (W-MaxLA) problem, which asks, given an undirected graph $G$ with weight function $f: V(G) \rightarrow \mathbb{Z}$ and positive integer $k$, if there is a linear arrangement $\pi$ with value at least $k$. The value of the arrangement $\pi$ is calculated as in Property \ref{prop:calculate value using edges cuts levels} using levels. And the level of the $i^{th}$ vertex $v_i$ in the arrangement $\pi$ is defined as 
\begin{equation}\label{eqn:level calc WMax}
    l_\pi(v_i) = f(v_i) - |N(v_i) \cap \{v_1, \dots, v_{i-1}\}|.
\end{equation}

The value of $c_i$ is defined similarly
$$c_i = \sum_{j=1}^i l_\pi(v_j).$$

Since the value of an arrangement is completely determined by the level of its vertices, the following is immediate.

\begin{property}\label{prop: Wmax equiv}
MaxDLA$(D,k)$ is equivalent to W-MaxLA$(G,f,k)$ when $G$ is the underlying graph of $D$ and $f$ maps the vertices of $G$ to their outdegree in $D$.
\end{property}

Finally, we have the following.

\begin{property}\label{prop: levels non-increasing}
The vertices of a maximal linear arrangement are arranged by levels, in non-increasing order.
\end{property}

\begin{proof}
Suppose, for a contradiction, $\pi = (v_1, v_2, \dots, v_n)$ is a maximal arrangement of a graph $G$ with weight function $f: V(G) \to \mathbb{Z}$ so that there are vertices $v_i$, $v_j$ where $i<j$ but $l_\pi(v_i) < l_\pi(v_j)$. We may assume, without loss of generality, that $j = i+1$, and we define a new arrangement $\pi'$ by interchanging $v_i$ and $v_j$ so that $\pi' = (v_1, v_2, \dots, v_{i-1}, v_j, v_i, v_{j+1}, \dots, v_n)$. It follows that $c_k(\pi) = c_k(\pi')$ for all $k \neq i$. If it is the case that $l_{\pi'}(v_j) > l_{\pi}(v_i)$, then it would mean $c_i(\pi') > c_i(\pi)$, which would contradict the maximality of $\pi$. But indeed this is the case: If $G$ has no edge $v_iv_j$, then $l_{\pi'}(v_j) = l_\pi(v_j) > l_\pi(v_i)$. And if $G$ has an edge $v_iv_j$, then $l_{\pi'}(v_j) = l_\pi(v_j) + 1 > l_\pi(v_i)$.
\end{proof}

\begin{property}\label{minValD_byAverage}
Let $D=(V,E)$ be a loopless multi-digraph with $|V| = n$ and $|E| = m$. Then $$MaxDLA(D) \geq \frac{1}{6}m(n+1).$$ This bound is best possible, and equality holds for complete symmetric digraphs.
\end{property}
\begin{proof}
Denote by $\Pi$ the set of all $n!$ possible arrangements of $D$, and consider the sum of their values 
\begin{equation}
val(\Pi) = \sum_{\pi \in \Pi}val(\pi).
\end{equation}
Let $i,j$ be integers, and $e=(u,v) \in E(D)$. For every $1 \leq i<j \leq n$, there are $(n-2)!$ arrangements of $D$ which map both $u$ to $i$ and $v$ to $j$. Therefore, each edge of $D$ contributes exactly $(n-2)!val(K_n)$ to $val(\Pi)$. It follows that 
\begin{align*}
val(\Pi) &= m(n-2)!val(K_n)\\
&= m(n-2)!\binom{n+1}{3}.
\end{align*}
Thus the mean value of all $n!$ arrangements of $D$ is $\frac{1}{n!}\cdot m(n-2)!\binom{n+1}{3} = \frac{1}{6}m(n+1)$. Since there must be an arrangement of $D$ whose value achieves the mean, this completes the proof of the inequality. 

We show that the bound is best possible by demonstrating that equality holds for complete symmetric digraphs. Indeed, the value of an arrangement of an $n$-vertex complete symmetric digraph is equal to $MaxLA(K_n)$, which is $\binom{n+1}{3}$. This is exactly the lower bound $\frac{1}{6}m(n+1)$ when $m = 2\binom{n}{2}$, the number of edges in a complete symmetric digraph.
\end{proof}

\section{NP-hardness results}

In this section we show that both MaxDiCut and MaxDLA are NP-complete, even when restricted to planar graphs with maximum degree 14. That MaxDiCut is NP-complete for planar graphs contrasts that the undirected version is polynomial-time solvable for planar graphs \cite{orlovdorf72, Hadlock1975}. At the end of the section we also show that MaxDLA is NP-complete for directed versions of split graphs, and as a corollary MinLA is NP-complete for split graphs.

A \emph{maximum directed cut} in a digraph $D$ is a directed cut $E(X,Y)$ with largest cardinality over all partitions $\{X,Y\}$. We denote the number of edges in such a cut by $maxDiCut(D)$.

A maximum linear arrangement of a digraph $D$ may or may not contain a maximum directed cut of $D$. However, $maxDiCut(D)$ and $maxDLA(D)$ are related by the following theorem, whose bounds are best possible.

\begin{theorem} If $D$ is a digraph with $n$ vertices and $maxDiCut(D) = t$, then $\frac{1}{2}nt \leq maxDLA(D) \leq (n-1)t$.
\end{theorem}

\begin{proof}
For the upper bound, observe that a linear arrangement of $D$ contains exactly $n-1$ cuts, none of which is larger than $t$. 

For the lower bound, let $E(X,Y)$ be a maximum directed cut of $D$, and construct $D'$ from $D$ by deleting all edges not in $E(X,Y)$. By construction, $D'$ is an orientation of a bipartite graph with all edges directed from $X$ to $Y$. We will now build an arrangement $\pi$ of $D'$ with value at least $\frac{1}{2}nt$. Let $\pi_x$ be an arrangement of $X$ by non-increasing outdegree in $D'$. Similarly, let $\pi_y$ be an arrangement of $Y$ by non-decreasing indegree in $D'$. Then $\pi = \pi_x, \pi_y$ is an arrangement of $D$ satisfying the lower bound. 
\end{proof}

The following corollary is immediate.

\begin{corollary}
The largest cut in a MaxDLA of digraph $D$ has size at least $$\frac{1}{2} MaxDiCut(D).$$
\end{corollary}

We now move on to our NP-Completeness results. First, we show that MaxDLA is NP-complete for orientations of planar graphs. We do this through two reductions: we reduce Planar Max 2SAT to Planar MaxDiCut, and then reduce MaxDiCut to MaxDLA in a way that preserves planarity. After this, we show that the problem remains NP-complete even with maximum degree 14.

Max2SAT asks, given a 2-CNF $\phi$ and positive integer $k$, if there is a truth assignment to the variables of $\phi$ so that at least $k$ clauses in $\phi$ are satisfied. Given a CNF $\phi$, its \emph{variable graph} is the graph with vertex set $V$ so that each $v \in V$ is associated with exactly one variable in $\phi$, and edge set $E$ so that $\{u,v\} \in E$ if and only if the variables associated with $u$ and $v$ occur together in some clause of $\phi$. Planar Max2SAT is an instance of Max2SAT whose 2-CNF has a variable graph that is planar. Planar Max2SAT is NP-complete \cite[p. 254]{guibas1991}.

\begin{theorem}\label{PlanarDiCut}
Planar MaxDiCut is NP-complete.
\end{theorem}
\begin{proof}

It is straightforward that Planar MaxDiCut is in NP.
We proceed by reduction from Planar Max2SAT. Consider an instance of Planar Max2SAT with boolean variables $U$, clauses $C$ where $|C|=l$, and integer $k$. We construct a planar digraph $D=(V,E)$ as follows (see Figure \ref{2SATtoMaxDiCut}):

\begin{enumerate}
    \item Let $V(D) = U$.
    \item For each $c \in C$ of size two, add a gadget consisting of four edges and one or two \lq gadget' nodes as follows:
    \begin{enumerate}
    \item If exactly one variable is positive in $c$, then connect both variables in $c$ by two internally disjoint paths of length two, oriented from the positive variable to the negative variable.
    \item Otherwise, add a pair of symmetric edges between the variables of $c$. Add one additional \lq gadget' vertex $v$ dominating the variables in $c$ so that $v$ is a sink when both variables of $c$ are positive, and a source when both variables are negative.
\end{enumerate}
    \item For each $c \in C$ of size 1, add a gadget consisting of two edges and two `gadget' nodes so that the edges connect the variable in $c$ to the gadget node. Orient the edges toward $c$ if the variable in $c$ is negative, and away from $c$ if the variable in $c$ is positive.
\end{enumerate}

Note that $D$ is planar because the instance of Max2SAT is planar.

\begin{figure}
    \centering
    \includegraphics[scale=0.9]{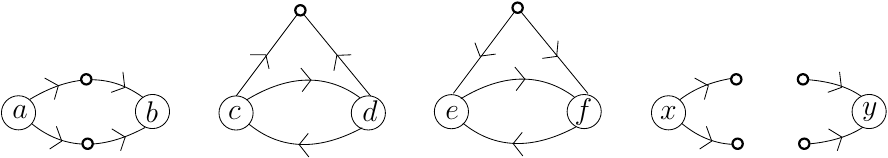}
    \caption{Gadgets for clauses $(a \lor \overline{b}), (c \lor {d}), (\overline{e} \lor \overline{f}), (x), (\overline{y})$.}
    \label{2SATtoMaxDiCut}
\end{figure}

\begin{claim}
There is a truth assignment to the variables of $U$ so that at least $k$ clauses in $C$ are satisfied if and only if $D$ has a dicut of size at least $2k$.
\end{claim}

\noindent Given such a truth assignment $\tau$, we construct a dicut $E=E(\mathcal{T},\mathcal{F})$ in $D$ as follows. First, set
\begin{align*}
    \mathcal{T} &= \{u \in U : \tau(u) = true\} \text{, and}\\
    \mathcal{F} &= \{u \in U : \tau(u) = false\} \text{,}
\end{align*}
and finally place the \lq gadget' vertices so the number of edges across $E$ is maximized. The last part is simple to do because they are all degree two. 

Now, because the gadgets are edge-disjoint, the contribution each makes to $E$ can be evaluated independently. It can be quickly verified that gadgets whose clause is true under $\tau$ contribute exactly two edges to $E$. Hence $E$ has size at least $2k$.

Conversely, suppose $D$ has a dicut $E(\mathcal{T},\mathcal{F})$ of size at least $2k$. We may assume the gadget vertices are arranged so as to maximize the size of this dicut. We find a truth assignment $\tau$ to the variables of $U$ as follows. For all $u \in U$, set
\begin{align*}
    \tau(u) &= true \text{ if } u \in \mathcal{T} \text{, and}\\
    \tau(u) &= false \text{ if } u \in \mathcal{F}.
\end{align*}
Observe that each gadget is constructed so that it contributes either two or zero edges to any dicut. The reader may quickly verify that (1) if a gadget contributes exactly two edges to $E$, then its clause is satisfied under $\tau$,
and (2) if a gadget contributes exactly zero edges to $E$, then its clause is not satisfied under $\tau$.

Since $E$ has size at least $2k$, and all edges of $D$ are in gadgets, it must be that at least $k$ gadgets contribute two edges to $E$, and so at least $k$ clauses of $C$ are satisfied under $\tau$. 
\end{proof}

We now strengthen Theorem \ref{PlanarDiCut} to include the restriction that no vertex has degree exceeding fourteen. We reduce from Planar 3SAT via Planar Max2SAT to Planar MaxDiCut, following existing reductions. However, we will be careful about limiting the number of clauses in which each variable occurs, and what those clauses look like. For this reason, we define a restricted version of Planar 3SAT. Given a CNF $\phi$, its bipartite \emph{clause-variable graph} is the graph with vertex set $V \cup C$ so that each $v \in V$ is associated with exactly one variable in $\phi$, each $c \in C$ is associated with exactly one clause in $\phi$, and edge set $E$ so that $\{v,c\} \in E$ if and only if the variable associated with $v$ occurs in the clause associated with $c$.

{Restricted Planar 3-Satisfiability (Restricted Planar 3SAT)} asks, given a CNF $\phi$ with 
a set $U$ of boolean variables, a set $C$ of clauses over $U$ so that the following conditions are satisfied:
\begin{enumerate}
    \item the clause-variable graph of $\phi$ is planar; and
    \item each variable in $U$ occurs in exactly three clauses of $C$ so that
    \begin{enumerate}
        \item two of the clauses contain exactly two literals each, with exactly one negated variable per clause,
        \item the third clause contains exactly three literals;
    \end{enumerate}
    
\end{enumerate}
{is $C$ satisfiable?}

Note that Restricted Planar 3SAT is in NP. We now prove that it is NP-complete following a reduction of \cite[pp. 96-97]{MaňuchJán2008FPCT} as done in \cite{tippenhauer2016}.

\begin{lemma}
Planar 3SAT is polynomial-time reducible to Restricted Planar 3SAT.
\end{lemma}

\begin{proof}
Consider an instance of Planar 3SAT $\phi$ with variables $U$, clauses $C$, and integer $k$. Find a planar embedding of the clause-variable graph in polynomial time \cite{AuslanderL.1961OIGi}. 

We construct an instance of Restricted Planar 3SAT $\psi$ from $\phi$ as follows. For every $u \in U$, let $l$ be the number of occurrences of $u$ in the clauses of $C$. Replace $u$ in $U$ with $l$ new variables $u_1, ..., u_l$ and, to $C$, add $l$ new clauses $\overline{u_i} \lor u_{i+1}$ (mod $l$) for $i = 1, ..., l$. Replace each occurrence of $u$ in an original clause of $C$ with one of the distinct new variables, choosing carefully so as to preserve planarity of the clause-variable graph, as in Figure~\ref{P3toRest}. Then $\psi$ is an instance of Restricted Planar 3SAT.

\begin{claim}
The original instance of Planar 3SAT $\phi$ is satisfiable if and only if the restricted version $\psi$ is satisfiable.
\end{claim}

Suppose $\phi$ is satisfiable. Then the same truth assignment $\tau$ which satisfies $\phi$ can be used to find a truth assignment which satisfies $\psi$. Just assign each new variable $u_i$ of $\psi$ a truth value equal to $\tau(u)$. 

Conversely, suppose $\psi$ is satisfiable. Let $\tau$ be a truth assignment which satisfies the clauses of $\psi$. Then it must be the case that if new variables $u_i$ and $u_j$ of $\psi$ stem from the same original variable $u$ of $\phi$, then $\tau(u_i) = \tau(u_j)$. This is because of the new clauses added to $\psi$ for variable $u$, which act like a circle of implication. So assign to each original variable $u$ of $\phi$ a truth value equal to $\tau(u_1)$. This is a truth assignment which satisfies the clauses of $\phi$.
\begin{figure}
    \centering
    \includegraphics[scale=0.8]{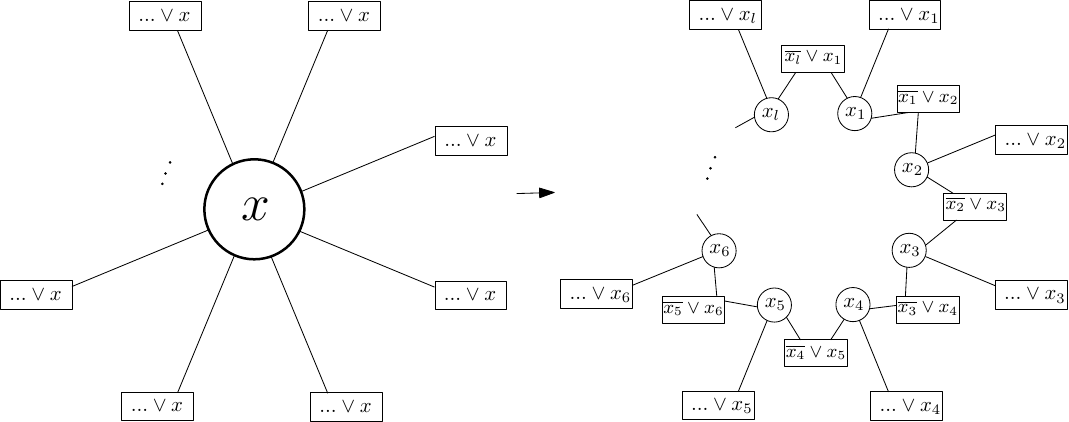}
    \caption{Planar 3SAT to restricted Planar 3SAT \cite{tippenhauer2016}}
    \label{P3toRest}
\end{figure}
\end{proof}

Now we define a restricted version of Planar Max2SAT.
The Restricted Planar Maximum 2-Satisfiability (Restricted Planar Max2SAT) problem asks, given a set $U$ of boolean variables, a set $C$ of clauses over $U$, and an integer $k$ so that the following conditions are satisfied:
\begin{enumerate}
    \item the variable graph of $C$ is planar;
    \item each $c \in C$ has $|c| \leq 2$; and
    \item each variable in $U$ occurs in at most six clauses of $C$, at most two of which contain either both positive or both negative variables;
\end{enumerate}
{is there a subset $C'$ of $C$ with $|C'| \geq k$ that is satisfiable?}

Note that Restricted Planar Max2SAT is in NP. We show that it is NP-complete following a reduction of Garey, Johnson, and Stockemeyer \cite[p. 240]{garey76}.



\begin{lemma}
Restricted Planar 3SAT is polynomial-time reducible to Restricted Planar Max2SAT.
\end{lemma}

\begin{proof}
Consider an instance of Restricted Planar 3SAT with variables $U$ and clauses $C$, where $C$ contains $l$ clauses of size 3, and $m=3l$ clauses of size 2. We construct an instance of Restricted Planar Max2SAT with variables $U'$, clauses $C'$, and integer $k$ as follows. First, set
\begin{align*}
    U' &= U \cup \{d_1, d_2, ..., d_l\}\text{, and}\\
    C' &= \{c \in C : c \text{ has size two}\}.
\end{align*}

Then add ten clauses to $C'$ for every clause $c_i = (a_i, b_i, c_i) \in C$, where $i = 1, ..., l$ of size three in the following way.
If $c_i$ contains all negated variables, add the following ten clauses to $C'$: 
    \begin{equation*}
    (a_i), (b_i), (c_i), (\overline{d_i}), (\overline{a_i} \lor \overline{b_i}), (\overline{b_i} \lor \overline{c_i}), (\overline{c_i} \lor \overline{a_i}), (a_i \lor {d_i}), (b_i \lor {d_i}), (c_i \lor {d_i}).
    \end{equation*}
Otherwise, add the following ten clauses to $C'$: 
    \begin{equation*}
    (a_i), (b_i), (c_i), (d_i), (\overline{a_i} \lor \overline{b_i}), (\overline{b_i} \lor \overline{c_i}), (\overline{c_i} \lor \overline{a_i}), (a_i \lor \overline{d_i}), (b_i \lor \overline{d_i}), (c_i \lor \overline{d_i}).
    \end{equation*}
Finally, let $k=7l + m$. The reader can verify that this is an instance of Restricted Planar Max2SAT.

\begin{claim}
Recall that $C$ contains $l$ clauses of size 3, and $m=3l$ clauses of size 2. The clauses $C$ are satisfiable if and only if $7l + m$ of the clauses in $C'$ are satisfiable.
\end{claim}
We begin with the crucial observation of Garey, Johnson and Stockemeyer \cite{garey76}, which the reader may verify. Let $\tau$ be a truth assignment to $U$. Then $\tau$ satisfies a clause $c \in C$ if and only if $\tau$ may be extended to $U'$ so that seven of the ten clauses in $C'$ stemming from $c$ are satisfiable. Further, if $\tau$ does not satisfy $c$, then any extension of $\tau$ to $U'$ can satisfy at most six of the ten clauses in $C'$ stemming from $c$.


Suppose $C$ is satisfiable. Let $\tau$ be a truth assignment to $U$ that satisfies $C$. By the observation above, $\tau$ may be extended to $U'$ so that seven of every ten clauses in $C'$ associated with a three-variable clause in $C$ are satisfied. Since this extension necessarily satisfies all $m$ 2-literal clauses in $C'$ that are also in $C$, the extension of $\tau$ satisfies $7l + m$ of the clauses in $C'$.

Conversely, suppose $7l + m$ of the clauses in $C'$ are satisfiable. Let $\tau'$ be a truth assignment to $U'$ that satisfies at least $7l +m$ of the clauses in $C'$, and let $\tau$ be $\tau'$ restricted to the variables in $U$. Observe that at most seven of the ten clauses in $C'$ stemming from a single 3-variable clause in $C$ can be simultaneously satisfied. Since there are only $l$ such collections of ten clauses, it must be the case that $\tau'$ satisfies seven of ten clauses for all such collections of ten clauses, and also all $m$ two-variable clauses in $C'$. But this means, by the observation above, that $\tau$ is a truth assignment to $U$ that satisfies $C$.
\end{proof}

We consider the degree of a vertex $v$ in a digraph to be the number of edges incident to $v$.

\begin{theorem}\label{maxdegee14}
MaxDiCut is NP-complete on planar graphs whose maximum degree is at most 14.
\end{theorem}
\begin{proof}
As previously observed, Planar MaxDiCut is in NP, therefore so is the restricted version.
We proceed by reduction from Restricted Planar Max2SAT, in the same manner as Theorem \ref{PlanarDiCut}. We will not repeat the details here, but only observe how the restricted version of Planar Max2SAT gives rise to the degree restrictions.

Consider an instance of Restricted Planar Max2SAT with boolean variables $U$, clauses $C$, and integer $k$. We construct a planar digraph $D=(V,E)$ as in Theorem \ref{PlanarDiCut}. 

Recall that, $V = U$, and the degree of any $v \in V$ depends only on the clauses $U$ that contain $v$. Since $v$ may be in at most six clauses, and all of those clauses but perhaps two (those with two variables of the same sign), contribute two to the degree of $v$, and the remaining contribute three, it is clear that no vertex may have degree exceeding 14.
\end{proof}

\begin{figure}[ht]
    \centering
    \includegraphics[scale=1.0]{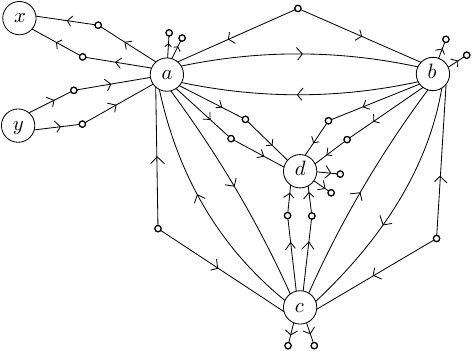}
    \caption{Construction of a reduction digraph for all clauses containing some variable $a$ in Restricted Planar 3SAT. The clauses in Restricted Planar 3SAT are $(a \lor b \lor c), (a \lor \overline{x}), (\overline{a} \lor y)$. The associated clauses in restricted Planar 2SAT are $(a), (b), (c), (d), (\overline{a} \lor \overline{b}), (\overline{b} \lor \overline{c}), (\overline{c} \lor \overline{a}), (a \lor \overline{d}), (b \lor \overline{d}), (c \lor \overline{d}); (a \lor \overline{x}); (\overline{a} \lor y)$.}
    \label{rest3SATtorestMaxDiCut}
\end{figure}

Figure \ref{rest3SATtorestMaxDiCut} depicts the digraph arising from the above reductions applied to a small instance of Restricted Planar 3SAT.


Theorem \ref{maxdegee14} prompts the question, what is the lowest maximum degree for which Planar MaxDiCut remains NP-complete? It is straightforward to observe that MaxDiCut is polynomial-time solvable on graphs whose maximum degree is two, and so the answer is somewhere between three and thirteen.

\begin{theorem}
MaxDLA is NP-complete on planar graphs whose maximum degree is at most 14.
\end{theorem}

\begin{proof}
We show that MaxDiCut is polynomial-time reducible to MaxDLA in a way that preserves maximum degree.
Consider an instance of MaxDiCut with digraph $D=(V,E)$ on $n$ vertices, and positive integer $k$. We may assume $D$ has at least one edge, and that $k \leq \binom{n}{2}$ because otherwise the solution to MaxDiCut is trivial.

Construct a new graph $D'$ from $D$ by adding a set $S$ of $n^3$ isolated vertices to $D$. That is, $D'=(V(D) \cup S, E(D))$. 
We will show that
$D$ has a dicut of size at least $k$ if and only if $D'$ has a linear arrangement of size at least $kn^3$.

Suppose $D$ has a dicut $C= E(A,B)$ of size at least $k$. Then, arrange the vertices of $D'$ so that all vertices in $A$ come before all vertices in $S$, which in turn come before all vertices in $B$. This arrangement has value at least $kn^3$ because the large cut $C$ is repeated for every vertex in $S$.

Now, suppose $G$ has no cut of size at least $k$. Then $G'$ also has no cut of size at least $k$. Hence an arrangement $\pi$ of $D'$ can have value at most $(k-1)(n^3 + n-1)$ because every one of the $(n^3 +n -1)$ cuts in $\pi$ has size at most $(k-1)$. But $k \leq \binom{n}{2} < n^2$ by assumption, so that $n^3 > (k-1)n$. But this means $\pi$ has value at most $(k-1)(n^3 + n -1) \leq kn^3 + (k-1)n - n^3 < kn^3$.
\end{proof}


\begin{corollary}\label{cor:S}
Let $\mathcal{C}$ be a class of digraphs closed under the operation of adding an isolated vertex. If MaxDiCut is NP-Complete for $\mathcal{C}$, then MaxDLA is also NP-Complete for $\mathcal{C}$.
\end{corollary}

By Corollary \ref{cor:S}, it follows that MaxDLA is NP-complete for both split digraphs and strict split digraphs. 

A \emph{split} graph is a graph whose vertex-set can be partitioned into an independent set and a clique.

A \emph{strict split} digraph, as defined in \cite{hellcruz17} is a digraph whose vertex set can be partitioned into an independent set and a strong clique. A \emph{strong clique} is a set of vertices $S$ in a digraph $D$ so that $(u,v) \in E(D)$ for every pair of distinct vertices $u,v \in S$.

A \emph{split} digraph, as defined in \cite{LaMarM.Drew2010Sd}, is a digraph whose vertex set can be partitioned into four sets, $A, B, C, D$ so that $A$ is a strong clique; $D$ is an independent set; all possible edges go from $A$ to $C$ and from $B$ to $A \cup C$; and no edges go from $D$ to $B$ or from $C$ to $B \cup D$. Additionally, it is forbidden to place all vertices in one of $B$ or $C$. It is straightforward to see that split graphs, strict split digraphs, and split digraphs are all closed under the operation of adding an isolated vertex.

\begin{corollary}
    MaxDLA is NP-complete for both split digraphs and strict split digraphs.
\end{corollary}

\begin{proof}
First, MaxCut for split graphs is NP-complete by a theorem of Bodlaender \cite{Bodlaender94}. Since the class of split graphs is contained in the intersection of strict split digraphs and split digraphs,\footnote{Every split graph, when each undirected edge is viewed as two symmetric directed edges, is both a split digraph and a strict split digraph.} it follows that MaxDiCut is NP-complete for those digraphs. By Corollary \ref{cor:S}, the result follows.
\end{proof}

\begin{corollary}
    MinLA is NP-complete for split graphs.
\end{corollary}
\begin{proof}
    Since MaxCut is NP-complete for split graphs, MaxLA (the undirected version of MaxDLA) is also NP-complete for split graphs by Corollary \ref{cor:S}. This means MinLA is NP-complete for the complement of split graphs, which is just the class of split graphs because they are closed under complementation.
\end{proof}


\section{Algorithm for Oriented Trees}
\label{ortrees}

In this section, we describe an algorithm solving MaxDLA on orientations of trees with degrees bounded by a constant. This same algorithm, with a slight modification described at the end of this section, solves MinLA on graphs $G$ when $\overline{G}$ is a bounded-degree tree.

The algorithm will work with \emph{sorted level signatures}, which are similar to signatures (of cuts, cf. Equation \ref{eqn: signature of cuts}), but useful in the algorithm to save computation time. We first define \emph{level signature}, and then \emph{sorted level signature}. 

The \emph{level signature} $l$ of an arrangement $\pi = (v_1, v_2, ..., v_n)$ is the $n$-tuple of the levels of its vertices, 
\begin{equation}
    l(\pi)=({l_1}, {l_2}, ..., {l_n}),
\end{equation}
where ${l_i}= l_\pi(v_i)$. The value of $l$ is defined in the natural way (cf. Equation \ref{eqn: level calculation}),
\begin{equation}\label{vallevelsig}
    val\big(l(\pi)\big) = \sum_{i=1}^{n} (n-i+1){l_i}
\end{equation}
so that $val\big(l(\pi)\big) = val\big(s(\pi)\big) = val(\pi)$.
The level signature of $\pi$ contains the same information as the signature of $\pi$, and it is straightforward to convert between the two. Indeed, if $s(\pi) = (c_1, c_2, ..., c_{n})$ then $l(\pi)=(c_1-0, c_2-c_1, ..., c_n-c_{n-1})$. 

The \emph{sorted level signature} $l^*$ of an arrangement $\pi = (v_1, v_2, ..., v_n)$ is obtained from $l(\pi)$ by sorting it into a non-decreasing order,
\begin{equation}
    l^*(\pi)=(l_1^*, l_2^*, ..., l_n^*),
\end{equation}
where $l_1^* \geq l_2^* \geq ... \geq l_n^*$. The value of $l^*$ is defined analogously to the value of $l$,
\begin{equation}
    val\big(l^*(\pi)\big) = \sum_{i=1}^{n} (n-i+1)l^*_i.
\end{equation}
Note that, in general, $val\big(l(\pi)\big)$ and $val\big(l^*(\pi)\big)$ may be different, however maximal arrangements are a special case.
Recall (Property \ref{prop: levels non-increasing}) that the vertices of a maximal linear arrangement are arranged by non-decreasing order of levels. Hence, if $\pi$ is a maximal arrangement, then $l(\pi) = l^*(\pi)$.

Our goal is to describe an algorithm solving MaxDLA on directed tree $D$ (whose maximum degree is bounded by a constant $d$), and integer $k$.
However, we describe instead an algorithm solving W-MaxLA for undirected tree $G$ (whose maximum degree is bounded by $d$), vertex weight function $f$ (which is also bounded by $d$), and integer $k$. The W-MaxLA problem contains the MaxDLA problem by Property \ref{prop: Wmax equiv}.

The algorithm is recursive. We make an observation which allows us to remove an edge from the tree, and hence break it apart, in \ref{findEdgeStep} and \ref{newFunctionsStep} of Algorithm \ref{boundedDegreeTreeAlg}.


\par

\begin{observation}\label{delEdge}
Let $G=(V,E)$ be an undirected graph, $e=uv \in E(G)$, $f:V(G) \mapsto \mathbb Z$ be a weight function, and $\pi$ be an arrangement of $V(G)$ with $\pi(u) < \pi(v)$. Define $G'= G-e$ and $f'$ as
\begin{equation}
      f'(x) = \begin{cases}
f(x) -1 &\text{if $x=v$}\\
f(x) &\text{otherwise.}
\end{cases}
  \end{equation} 
  Then the signature of $\pi$ is unchanged under either $G,f$ or $G', f'$, that is $s_{G,f}(\pi) = s_{G',f'}(\pi)$. 
\end{observation}
\begin{proof}
By Equation \ref{eqn:level calc WMax} it is straightforward to see that the level of a vertex $x$ in $\pi$ is unchanged under either $G,f$ or $G', f'$, that is $l_{G,f}(x) = l_{G',f'}(x)$.
\end{proof}

We describe an algorithm which, given a bounded-degree undirected tree $G$, and bounded weight function $f$ on its vertices, outputs a set of sorted level signatures of $G,f$. Figure \ref{alg} shows an example flow of this algorithm for a path of length three.

\begin{algorithm}\label{boundedDegreeTreeAlg}
FindSignatures($G$,$f$)
\end{algorithm}

\quad INPUT: A tree $G=(V,E)$ on $n$ vertices with maximum degree $d$, and a weight function $f$ from $V(G)$ to the integers, where $f(v) \leq d$, with $d$ a constant.
\par
\quad OUTPUT: A set $S$ of sorted level signatures of $G,f$, so that the largest value of an element in $S$ is equal to $W\text{-}MaxDLA(G,f)$.

\begin{steps}
  \item \label{base step}If $G$ is an isolated vertex $v$, set $S = \{ (f(v))\}$ and terminate.
  \item \label{findEdgeStep}If $G$ is not an isolated vertex, find an edge $e=(u,v)$ that minimizes the size of the largest component when $e$ is deleted.
  \item\label{newFunctionsStep} Define two new weight functions $f_u$ and $f_v$ so that  
  \par
  \begin{equation}
      f_u(x) = \begin{cases}
f(x) -1 &\text{if $x=u$}\\
f(x) &\text{otherwise}
\end{cases}
  \end{equation} \qquad and \qquad \begin{equation}
      f_v(x) = \begin{cases}
f(x) -1 &\text{if $x=v$}\\
f(x) &\text{otherwise.}
\end{cases}
  \end{equation}
  \item \label{recursionStep}Recursively calculate the sorted level signatures of $G-e$ with both $f_u$ and $f_v$. Do this in the following way. Let $G_1$ and $G_2$ be the two components of $G-e$. Make four recursive calls as follows: 
  \begin{align*}
      S_{1,u} &= FindSignatures(G_1,f_u|_{V(G_1)}),\\
      S_{2,u} &= FindSignatures(G_2,f_u|_{V(G_2)}),\\
      S_{1,v} &= FindSignatures(G_1,f_v|_{V(G_1)}),\text{ and}\\
      S_{2,v} &= FindSignatures(G_2,f_v|_{V(G_2)}).
  \end{align*}
  
  \item \label{combine}Combine $S_{1,u}$ and $S_{2,u}$ to form $S_u$ by combining every sorted level signature in $S_{1,u}$ with every sorted level signature in $S_{2,u}$ as follows. If $l = (l_1, l_2, ..., l_p) \in S_{1,u}$ and $l' = (l_1', l_2', ..., l_q') \in S_{2,u}$ then their new combined sorted level signature is obtained from $(l_1, l_2, ..., l_p,l_1', l_2', ..., l_q')$ by sorting it into a non-increasing order. Delete duplicate signatures. Do the same for $S_{1,v}$ and $S_{2,v}$ to form $S_v$.
  \item \label{unionStep}Let $S = S_u \cup S_v$.
  \item \label{return}Return $S$.
\end{steps}

On completion of Algorithm \ref{boundedDegreeTreeAlg}, it is simple to check if there is a sorted level signature in the output $S$ with value at least $k$.

\subsection*{Correctness}
We analyze the correctness of Algorithm \ref{boundedDegreeTreeAlg}.

\begin{lemma}\label{new}
    Let $S$, $G$, and $f$ be as described in Algorithm \ref{boundedDegreeTreeAlg}. Every element in the output $S$ is a sorted level signature of $G,f$.
\end{lemma}
\begin{proof}
    We proceed by induction on $|V(G)|$. As a base case, if $V = \{v\}$, then by Step 1 $S =\{(f(v))\}$. But the level of $v$ in the (unique) arrangement $\pi=(v)$ is $f(v)$ by Equation \ref{eqn:level calc WMax}. It follows that $(f(v))$ is a sorted level signature of $G,f$.

    For the inductive step, suppose $|V| > 1$, and suppose $l \in S$. Let $e = uv$ be the edge chosen in Step 2, and let $G_1, G_2$ be as in Step 4, with $v \in V(G_1)$. We may suppose without loss of generality that $l \in S_u$ in Step 5 (intuitively this means the algorithm is ``guessing'' that $v$ comes before $u$); and that $l_1 \in S_{1,u}$ and $l_2 \in S_{2,u}$ are two sorted level signatures which were combined in Step 5 to make $l$.

    By the induction hypothesis, $l_1$ and $l_2$ are sorted level signatures of $G_1,f_u|_{V(G_1)}$ and $G_2,f_u|_{V(G_2)}$ respectively. Suppose then, that $l_1$ is a sorted level signature of the arrangement $(v_{1_1}, v_{1_2}, \dots, v_{1_p})$, while $l_2$ is a sorted level signature of the arrangement $(v_{2_1}, v_{2_2}, \dots, v_{2_q})$. But, by Equation \ref{eqn:level calc WMax}, the level of a vertex in an arrangement depends only on its placement relative to its neighbours. Thus it follows from Observation \ref{delEdge}, and because $l \in S_u$, that $l$ is a sorted level signature of the arrangement $(v_{1_1}, v_{1_2}, \dots, v_{1_p},v_{2_1}, v_{2_2}, \dots, v_{2_q})$.
\end{proof}
We must show that the maximum value of a sorted level signature in $S$ is equal to $W\text{-}MaxDLA(G,f)$.

\begin{theorem}
Let $S$ be the final set returned by Algorithm \ref{boundedDegreeTreeAlg} when tree $G$ and weight function $f$ are input to the algorithm. Then W-MaxDLA$(G,f)$ is the maximum value of a sorted level signature in $S$.
\end{theorem}

\begin{proof}
First, we define the notion of \emph{consistent}. If $e = uv \in E(G)$ is chosen in Step 3 of the algorithm, and $f_u,f_v$ are as defined in that step, then an arrangement $\pi$ of $G,f$ is consistent with $f_u, G-e$ if $\pi(v) < \pi(u)$, and consistent with $f_v, G-e$ if $\pi(u) < \pi(v)$. For input $G,f$, we define a \emph{path} in the recursion of the algorithm as a sequence of successive recursive calls (which recurse on either both components of of $G-e$, with  $f_u$ or both components of $G-e$ with  $f_v$). Intuitively, one may think of a path in the recursion as a sequence of choices for ordering the ends of each edge. A path in the recursion of the algorithm is \emph{consistent} with an arrangement $\pi$ if every recursive call in that path is made on a graph, weight function which is consistent with $\pi$. 

Denote the maximum value of a sorted level signature in $S$ by ${l}^{max}$.
First, we show that $l^{max} \geq W\text{-}MaxDLA(G,f)$.
Let $\pi$ be an arrangement of $G,f$ so that $val(\pi) = W\text{-}MaxDLA(G,f)$. 
Let $l^* \in S$ be a sorted level signature obtained by following a path in the recursion that is consistent with $\pi$. By Equation \ref{eqn:level calc WMax} and the definition of consistent, it follows that for every $v \in V(G)$ the level of $v$ in $l^*$ is equal to its level in $\pi$. Since the vertices in $\pi$ are arranged by non-increasing levels per Property \ref{prop: levels non-increasing}, it follows that the value of $l$ is equal to the value of $\pi$. 
Therefore $l^{max} \geq W\text{-}MaxDLA(G,f)$.

Next, we show that  $l^{max} \leq {W\text{-}MaxDLA(G,f)}$. Let $l^* = (l_1^*, l_2^*, ..., l_n^*) \in S$ be a sorted level signature whose value is $l^{max}$. For $i=1, \dots, n$, let $v_i$ be the vertex associated with $l_i^*$. That is, at the base of the recursion in \ref{base step}, $v_i$ has level $l_i^*$. 

We claim that $l^*$ is obtained by following a recursive path that is consistent with $\pi = (v_1, v_2, \dots, v_n)$. Suppose the claim is false, then there exists an edge $e = v_iv_j \in E(G)$ with $i < j$, so that a path $P$ in the recursion resulting in $l^*$ is consistent with an arrangement where $v_j$ comes before $v_i$.

Let $P'$ be a path in the recursion making identical choices to $P$ at every edge except edge $e$, so that $P'$ chose instead $v_j$ before $v_i$. Let ${l^*}'$ be the sorted level signature resulting from $P'$. Then the level of every vertex is the same in both ${l^*}$ and ${l^*}'$, except $v_i$ and $v_j$. The level of $v_i$ is one greater, and the level of $v_j$ is one less, in ${l^*}$ than in ${l^*}'$. But since $v_i$ precedes $v_j$ in $l^*$, it means $l_i^* \geq l_j^*$. It follows that $val({l^*}') > val(l^*)$, which contradicts that $val({l^*}) = l^{max}$.

This means that $val(l^*) = val(\pi)$. Since $val(\pi) \le W\text{-}MaxDLA(G,f)$, it follows that $l^{max} \leq {W\text{-}MaxDLA(G,f)}$.
\end{proof}

\subsection*{Running Time}
We show that the running time is $O(n^{4d})$.

\ref{findEdgeStep} can be done in time linear in $n$ by, for example, moving from the leaves upward and orienting each edge of the tree toward the largest subtree when that edge is deleted. \ref{newFunctionsStep} can be done in constant time.

In \ref{recursionStep} of Algorithm \ref{boundedDegreeTreeAlg}, four recursive calls are made, and the number of vertices in the tree for these calls has size at most $\frac{(n-1)(d-1)}{d}+1$ where $n=|V(G)|$. This means that at each step of recursion the size of the problem (in terms of number of vertices) is reduced by a factor of at least $\frac{d}{d+1}$. 

\ref{combine} and \ref{unionStep} take up the majority of computation time. But this remains polynomial because we are only keeping track of sorted level signatures. A sorted level signature is concerned only with the number of vertices at each level. There are $2d+1$ possible levels, so only $O(n^{2d})$ possible signatures. This means at \ref{combine}, that  $|S_{1,u}|$ and $|S_{2,u}|$ are $O(n^{2d})$ and so can be combined in time $O(n^{4d})$. At \ref{unionStep}, again, $|S|$ is $O(n^{2d})$, so this step can also be done in time $O(n^{4d})$. The bound of $O(n^{4d})$ dwarfs the time required to sort and delete duplicates in \ref{combine} and \ref{unionStep}.

This leads to a recurrence of $T(n) = 4T(\frac{d}{d+1}n) + O(n^{4d})$, and running time $O({n^{4d}})$. Therefore, the running time of Algorithm \ref{boundedDegreeTreeAlg} is $O(n^{4d})$.

\begin{figure}[ht]
    \centering
    \includegraphics[scale=.92]{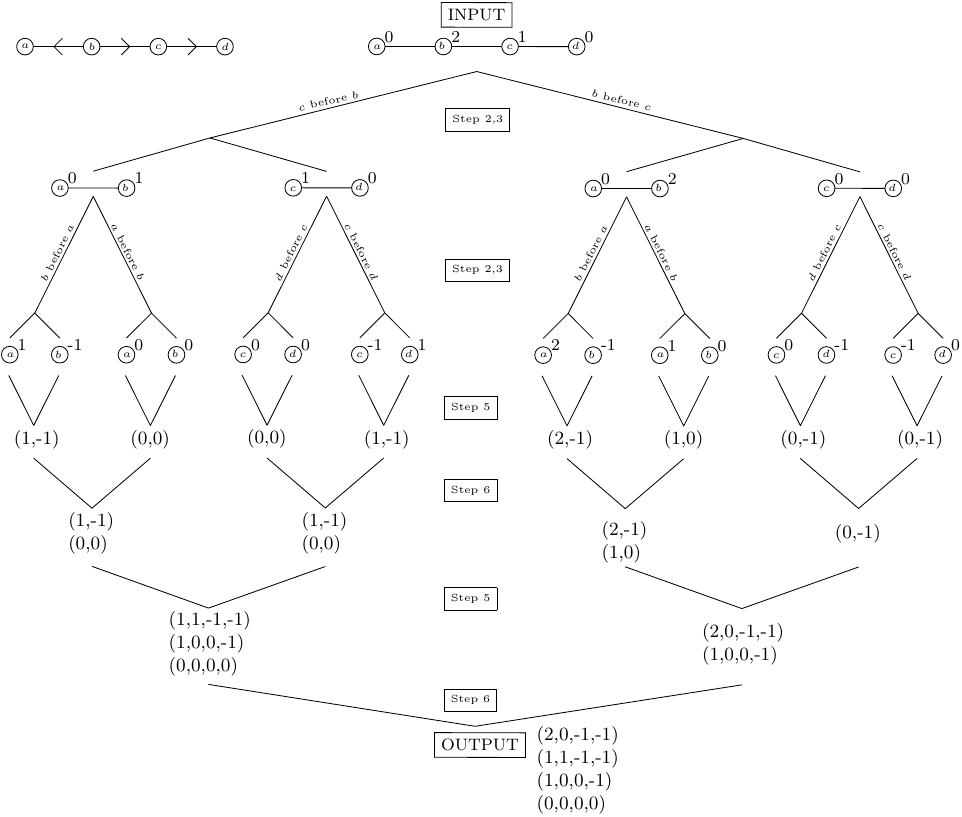}
    \caption{Flow of FindSignatures $(G,f)$ when solving for the above orientation of a path of length three.}
    \label{alg}
\end{figure}

\subsection*{Modification to solve MinLA}

Algorithm \ref{boundedDegreeTreeAlg} can be modified to show the following.

\begin{theorem}
MinLA is polynomial-time solvable on graphs $G$ where $\overline{G}$ is a tree with degree bounded by a constant.
\end{theorem}

As observed in the introduction, solving MaxLA for bounded-degree undirected trees is equivalent to solving MinLA on their complements. 
The modification is simple. First, we change Observation \ref{delEdge} by allowing graph $G$ to be a multigraph, with $k$ edges between $u$ and $v$. 

\begin{observation}\label{delEdgeUndir}
Let $G=(V,E)$ be an undirected multigraph, $e=uv \in E(G)$ with multiplicity $k$, $f:V(G) \mapsto \mathbb Z$ be a weight function, and $\pi$ be an arrangement of $V(G)$ with $\pi(u) < \pi(v)$. Define $G'= G-e$ and $f'$ as follows:
\begin{equation}
      f'(x) = \begin{cases}
f(x) -k &\text{if $x=v$}\\
f(x) &\text{otherwise.}
\end{cases}
  \end{equation} 
Then the signature of $\pi$ is unchanged under either $G,f$ or $G', f'$, that is $s_{G,f}(\pi) = s_{G',f'}(\pi)$.
\end{observation}

Now, to solve MinLA on an undirected simple graph $G$ where $\overline{G}$ is a bounded-degree tree, and integer $k$, use Algorithm \ref{boundedDegreeTreeAlg} with the following two modifications:
\begin{enumerate}
    \item Input $\overline{G},f$ where $f:V(G) \mapsto \mathbb{Z}$ so that $f(v)= d_{\overline{G}}(v)$.
    \item In \ref{newFunctionsStep}, decrement the weight function by two instead of one.
\end{enumerate}

The first modification is based on, as previously mentioned, the fact that an arrangement is minimum for $G$ exactly when it is maximum for $\overline{G}$. The second modification follows from Observation \ref{delEdgeUndir} by considering $\overline{G}$ as a symmetric digraph. 

It follows that $G$ has an arrangement of size at most $k$ if and only if one of the level signatures output from modified Algorithm \ref{boundedDegreeTreeAlg} has value\footnote{The value of an arrangement of a complete graph on $n$ vertices $K_n$ is $\sum_{i=1}^{n} i(n-i)=\binom{n+1}{3}$.} at least $\binom{n+1}{3} - k$.

\section{Graphs with a Maximum Arrangement}\label{sect:wonderful}

Inspired by a theorem of Harper exhibiting cuts of the hypercube minimum for every cardinality \cite{Harper1966}, we look for graphs with an arrangement so that every cut is \emph{maximum} for its position in the arrangement. 

First, we make these notions explicit. Recall the partial order $S$, coordinatewise on cut signatures, defined in Section \ref{sect:properties}. We have thus far considered maximal signatures, but in this section we will consider digraphs $D$ (or graphs $G$ with weight function $f$) for which the partial order $S$ has a greatest (or least) element. We call such digraphs (or graphs with weight function) \emph{nested-maximum} and \emph{nested-minimum} respectively. We call an arrangement of $D$ (or $G,f$) whose signature is the greatest (least) element of $S$ a \emph{nested-maximum} (respectively \emph{nested-minimum}) arrangement. In this language, Harper's theorem \cite{Harper1966} is that a hypercube is nested-minimum.

Three classes of nested-maximum digraphs are presented: tournaments, orientations of graphs with maximum degree at most two, and transitive acyclic digraphs. Tournaments and transitive acyclic digraphs are also nested-minimum. 

\subsection*{Tournaments}

Recall that a tournament is an orientation of a complete graph. We show that a tournament is both nested-maximum and nested-minimum.


\begin{theorem}
Let $D$ be a tournament on $n$ vertices with out-degree sequence $d_1^+ \geq d_2^+ \geq ... \geq d_n^+$. Then $D$ is both nested-maximum and nested-minimum. Furthermore, the maximum value of an arrangement of $D$ is 
\begin{equation}\label{maxTourney}
MaxDLA(D) = \sum_{i=1}^{n-1}(n-i)d_i^+ - \binom{n}{3},
\end{equation}
and it is achieved when $D$ is arranged by non-increasing out-degree. The minimum value of an arrangement of $D$ is 
\begin{equation}\label{minTourney}
MinDLA(D) = \sum_{i=1}^{n-1}id_i^+ - \binom{n}{3},
\end{equation}
and it is achieved when $D$ is arranged by non-decreasing out-degree.
\end{theorem}
\begin{proof}
We abstract to W-MaxLA, and solve the equivalent problem of arranging the complete graph on $n$ vertices $K_n$ with vertex weight function $f(v) = d^+(v)$. 

Let $\pi = (u_1,u_2, ..., u_n)$ be any arrangement of $K_n,f$, and let $\pi^+ = (v_1, v_2, ..., v_n)$ be an arrangement of $K_n,f$ by non-increasing out-degree in $D$, that is, so that $f(v_i) \geq f(v_j)$ whenever $i \leq j$. We will show that the $i^{th}$ cut of $\pi$ is no larger than the $i^{th}$ cut of $\pi^+$. That is, we will show that $c_i(\pi) \leq c_i(\pi^+)$.

Observe that the value of the $i^{th}$ cut of $\pi$ is 
\begin{equation*}c_i(\pi)= \sum_{j = 1}^il_\pi(u_j) = \sum_{j = 1}^i(f(v_j) - |N(v_j) \cap S_j|)
\end{equation*}

But $|N(v_j) \cap S_j| = j-1$ because we are arranging a complete graph. Therefore it follows that 

\begin{equation*}c_i(\pi)= \sum_{j=1}^if(v_j) - \frac{1}{2} i(i-1)
\end{equation*}

Now it is easy to see that $\sum_{j=1}^if(v_j) = \sum_{j=1}^id^+(v_j) \leq \sum_{j=1}^id^+_j$ and so 

\begin{equation*}c_i(\pi) \leq \sum_{j=1}^id^+_j - \frac{1}{2} i(i-1) = c_i(\pi^+) .
\end{equation*}

Thus arranging $D$ by non-increasing out-degree achieves the best possible value for every cut. Hence $D$ is nested maximum, and the value of a maximum arrangement of $D$ is given by the sum of its cuts\footnote{Note that we are using the identities $\sum _{i=1}^{n-1}i^2 = \binom{n}{2} + 2\binom{n}{3}$, and $\sum _{i=1}^{n-1}i = \binom{n}{2}$.}, $$MaxDLA(D) = \sum_{i=1}^{n-1}c_i = \sum_{i=1}^{n-1}\Bigg(\sum_{j=1}^id^+_j - \frac{1}{2} i(i-1)\Bigg) = \sum_{i=1}^{n-1}(n-i)d_i^+ - \binom{n}{3}.$$
Because the complement of a tournament is also a tournament, it follows that a tournament is also nested-minimum, and a minimum arrangement of a tournament is achieved when arranged by non-decreasing out-degree. The value of a minimum arrangement is calculated in a similar way to the calculation of a maximum arrangement. 
\end{proof}

It is interesting to note that, for an an Eulerian tournament $T$ on $n$ vertices, every arrangement of $T$ has value $MaxDLA(T) = MinDLA(T) = \frac{1}{2}m(n+1)$. 
Which is (necessarily) equal to the lower bound from Property \ref{minValD_byAverage} on the value of any $n$-vertex, $m$-edge loopless multi-digraph. Therefore, of all the loopless multi-digraphs on $n$ vertices and $m=\binom{n}{2}$ edges, an Eulerian tournament achieves the worst possible value in terms of MaxDLA.
\subsection*{Orientations of Graphs with Maximum Degree Two}


We show that orientations of graphs with maximum degree at most two are nested-maximum. We do this by proving something stronger, that for any undirected graph $G$ with maximum degree at most two, and any weight function $f: V(G) \mapsto \mathbb{Z}$, then $G,f$ has a nested-maximum arrangement. By Proposition \ref{prop: Wmax equiv}, this means that any orientation of $G$ is nested-maximum.

We begin with a lemma showing a condition which, when satisfied, means it is safe to assume one end of an edge comes before the other
in a maximal arrangement of $G,f$.

\begin{lemma}\label{edgePutEndFirst}
Let $G=(V,E)$ be a digraph and let $f: V(G) \mapsto \mathbb{Z}$ be a weight function. If there exists an edge $e=uv \in E(G)$ so that 
\begin{equation}\label{putUbeforeV}
    d(u) -1 \leq f(u) - f(v),
\end{equation}
then for any arrangement $\sigma$ of $G$, there is an arrangement $\pi$ of $G$ so that $\pi(u) < \pi(v)$ and $\sigma \leq \pi$.
\end{lemma}

\begin{proof}
Let $\sigma$ be an arrangement of $G$. We may assume $\sigma$ is a maximal arrangement, and so, by Property \ref{prop: levels non-increasing}, the vertices are arranged by non-increasing level. If $\sigma(u) < \sigma(v)$ then $\sigma$ satisfies the lemma, and we are done. So we may assume $\sigma(v) < \sigma(u)$, and by maximality of $\sigma$, that $l_\sigma(v) \geq l_\sigma(u)$. But if $l_\sigma(v) = l_\sigma(u)$ we may simply interchange $u$ and $v$ in $\sigma$ without altering any cut of $\sigma$. So it must be that 
\begin{equation}\label{v>u}
    l_\sigma(v) > l_\sigma(u).
\end{equation}
But now,
\begin{equation}
    l_\sigma(u) \geq f(u)-d(u) \geq f(v) -1 \geq l_\sigma(v) -1 \geq l_\sigma(u).
\end{equation}
This means 
\begin{equation}
    l_\sigma(v) -1 = l_\sigma(u).
\end{equation} We may permute vertices at the same level in $\sigma$ so that 
\begin{equation}
    \sigma(v) = \sigma(u) -1.
\end{equation}
Now simply interchanging $u$ and $v$ in $\sigma$ doesn't alter any cut. Observe that when $u$ and $v$ are interchanged, their levels are interchanged as well. This new arrangement satisfies the requirements of the lemma.
\end{proof}

\begin{theorem}
Let $G=(V,E)$ be a graph with $\Delta(G) \leq 2$, and $f:V(G) \mapsto \mathbb{Z}$ be a weight function. Then $G$ is nested-maximum.
\end{theorem}
\begin{proof}
We proceed by induction on $m = |E(G)|$. As a base case, assume $|E(G)| = 0$. Then $G$ is just isolated vertices, and so every arrangement of $G$ is necessarily maximum. For the inductive step, assume $|E(G)| \geq 1$. 

If $G$ is not connected, we may take a nested-maximum arrangement of each component of $G$ by induction. By Observation \ref{obs: disconnected}, any maximum arrangement of $G$ restricted to a component of $G$ must be a nested-maximum arrangement of that component. Since, by Property \ref{prop: levels non-increasing}, the levels of a maximal arrangement of $G$ are non-increasing, this means there is only one signature of a maximal arrangement of $G$. Therefore $G$ is nested-maximum. Hence we may assume $G$ is connected, and so $G$ is either a cycle or a path.

Say an edge $e=uv \in E(G)$ is \emph{easy} if for any arrangement $\sigma$ of $G$, there is an arrangement $\pi$ of $G$ so that $\pi(u) < \pi(v)$ and $\sigma \leq \pi$.
First, we show that $G$ always an easy edge.

\emph{Case 1}: Suppose for all edges $e=uv \in E(G)$, it is the case that $f(u) = f(v)$. Then, because $G$ is connected, $f$ is constant. If $G$ is a cycle, then by symmetry every edge is easy. If $G$ is not a cycle, then $G$ is a path. Let $u$ be an end of $G$ and $v$ be the neighbour of $u$. We claim that $e = uv$ is easy. Indeed, $f(u) = f(v)$ by assumption, but $d(u) = 1$. So $d(u) - 1 = 0 = f(u) - f(v)$, and Lemma \ref{edgePutEndFirst}, shows $e$ is easy.

\emph{Case 2}: There exists an edge $e=uv \in E(G)$ such that $f(u) > f(v)$. Then, because $d(u) \leq 2$ it follows that $d(u) - 1 \leq 1 \leq f(u) -f(v)$. Again, by Lemma \ref{edgePutEndFirst} $e$ is easy.

Therefore, in all cases, $G$ has an easy edge $e=uv$, and it is the case that
\begin{equation}\label{fuLessfv}
    f(u) \geq f(v).
\end{equation} 
Now construct a new graph $G' = G-e$, and associate with $G'$ a new weight function $f':V(G) \mapsto \mathbb{Z}$ where 
\begin{equation}\label{functionf'}
      f'(x) = \begin{cases}
f(x) -1 &\text{if $x=v$}\\
f(x) &\text{otherwise.}
\end{cases}
  \end{equation}

By this construction, and Observation \ref{delEdge}, for any arrangement $\pi$ of $G$, where $\pi(u) < \pi(v)$, it is the case that $s_{G,f}(\pi) = s_{G',f'}(\pi)$.


By the induction hypothesis, $G'$ with any weight function is nested-maximum, so let $\pi'$ be a nested-maximum arrangement of $G',f'$. We claim that $l_{\pi'}(u) \geq l_{\pi'}(v)$. Indeed, 
\begin{equation}
    l_{\pi'}(u) \geq f'(u) -d(u)_{G'} \geq f'(u) - 1 \geq f'(v) \geq l_{\pi'}(v),
\end{equation}
where the first and last inequalities follow from the definition of weighted level, the second because $d(u)_{G'} \leq 1$ as $e$ was deleted, and the third from Equations \ref{fuLessfv} and \ref{functionf'}. But this means that we may assume that $\pi'(u) < \pi'(v)$ because the vertices of a maximal arrangement are ordered by non-increasing level, and we may permute vertices within a level.

Now we claim that $\pi'$ is a nested-maximum arrangement of $G$. Suppose not, then there exists an arrangement $\pi''$ so that $\pi''$ beats $\pi'$ on some cut. Since edge $e=uv$ is easy, we may assume $\pi''(u) < \pi''(v)$. But by Observation \ref{delEdge} this means that $s_{G,f}(\pi'') = s_{G',f'}(\pi'')$. Because $\pi'$ is nested-maximum for $G',f'$, it follows that $\pi''$ cannot beat $\pi'$ on any cut. Therefore $\pi'$ is also nested-maximum for $G$.
\end{proof}

\begin{corollary}
Orientations of graphs $G$ where $G$ has maximum degree at most two are nested-maximum.
\end{corollary}

\subsection*{Transitive Acyclic Digraphs}

Finally, we show that every transitive acyclic digraph is both nested-minimum and nested-maximum. Recall that a digraph $D=(V,E)$ is {transitive} when the relation $E$ is transitive. Trivially, transitive acyclic digraphs are nested-minimum because an arrangement where every edge is backward is possible. We show that they are also nested-maximum. Specifically, we show that a nested-maximum arrangement of a transitive acyclic digraph is obtained by ordering the vertices $v$ by non-increasing $d^+(v) - d^-(v)$. 

\begin{theorem}
Let $D=(V,E)$ be a transitive acyclic digraph, and $\pi$ an arrangement of $V(D)$ by non-increasing $d^+(v) - d^-(v)$. Then  $\pi$ is nested-maximum. 
\end{theorem}
\begin{proof}
First we show that every maximal arrangement of $D$ is a topological sort, meaning every edge points forward. We proceed by contrapositive. Let $\pi = (v_1, v_2, ..., v_n)$ be an arrangement of $D$ with a backward edge $e=(v_j, v_i)$. Choose $e$ so that $j-i$ is minimum. Then, by transitivity, $v_i$ and $v_j$ must have no neighbour $v_k$ with $i < k < j$, otherwise our choice would not be minimum. But this means the arrangement $\pi'$ obtained from $\pi$ by interchanging $v_i$ and $v_j$, that is $\pi' = (v_1, ..., v_{i-1}, v_j, v_{i+1}, ..., v_{j-1}, v_i, v_{j+1}, ..., v_n)$, is strictly better than $\pi$. Indeed, cuts $C_i$ to $C_{j-1}$ are strictly increased in $\pi'$ because they include the edge $e$, and all other cuts have the same value in both $\pi$ and $\pi'$. Hence $\pi$ is not maximal, and therefore every maximal arrangement of $D$ is a topological sort.

But this means the level of a vertex $v$ in a maximal arrangement is $d^+(v) - d^-(v)$. Since vertices are arranged by non-increasing level in a maximal arrangement, it follows that every maximal arrangement has the same signature. The theorem follows.
\end{proof}

\bibliographystyle{acm}

\bibliography{bib}

@book{Diestel3rdEdition,
series = {Graduate texts in mathematics},
publisher = {Springer},
volume = {173},
isbn = {3540261826},
year = {2005},
title = {Graph {T}heory},
edition = {3rd},
author = {Diestel, Reinhard},
lccn = {2005928165},
}

@techreport{goldklip,
	Author = {M. K. Goldberg and I. A. Klipker},
	Institution = {Physico-Technical Institute of Low Temperatures, Academy of Sciences of Ukranian SSR},
	Title = {Minimal placing of trees on a line},
	Year = {1976}}

@article{AuslanderL.1961OIGi,
issn = {00959057},
journal = {Journal of Mathematics and Mechanics},
pages = {517--523},
volume = {10},
publisher = {The Graduate Institute for Mathematics and Mechanics INDIANA UNIVERSITY},
number = {3},
year = {1961},
title = {On Imbedding Graphs in the Sphere},
language = {eng},
author = {Auslander, L. and Parter, S. V.},
}

@InProceedings{Bodlaender94,
author="Bodlaender, Hans L.
and Jansen, Klaus",
editor="Enjalbert, Patrice
and Mayr, Ernst W.
and Wagner, Klaus W.",
title="On the complexity of the maximum cut problem",
booktitle="STACS 94",
year="1994",
publisher="Springer Berlin Heidelberg",
address="Berlin, Heidelberg",
pages="769--780",
isbn="978-3-540-48332-8"
}

@article{chung84,
	Author = {F. R. K. Chung},
	Journal = {Computers and Mathematics with Applications},
	Number = {1},
	Pages = {43-60},
	Title = {On optimal linear arrangements of trees},
	Volume = {10},
	Year = {1984}}

@article{cohen06,
	Author = {Cohen, J and Fomin, F and Heggernes, P and Kratsch, D and Kucherov, G},
	Issn = {0302-9743},
	Journal = {Mathematical Foundations Of Computer Science 2006, Proceedings},
	Language = {English},
	Pages = {267--279},
	Publisher = {SPRINGER-VERLAG BERLIN},
	Title = {Optimal linear arrangement of interval graphs},
	Volume = {4162},
	Year = {2006}}

@article{diaz02,
	Author = {Josep. Diaz and Jordi Petit and Maria Serna},
	Journal = {ACM Computing Surveys},
	Number = {3},
	Pages = {313 - 356},
	Title = {A Survey of Graph Layout Problems},
	Volume = {34},
	Year = {2002}}

@article{Diaz2007,
	Author = {Diaz, Josep and Kaminski, Marcin},
	Issn = {0304-3975},
	Journal = {Theoretical Computer Science},
	Language = {eng},
	Number = {1},
	Pages = {271--276},
	Publisher = {Elsevier B.V.},
	Title = {{MAX-CUT} and {MAX-BISECTION} are {NP}-hard on unit disk graphs},
	Volume = {377},
	Year = {2007}}

@article{garey76,
	Author = {Garey, M.R. and Johnson, D.S. and Stockmeyer, L.},
	Issn = {0304-3975},
	Journal = {Theoretical Computer Science},
	Language = {eng},
	Number = {3},
	Pages = {237--267},
	Publisher = {Elsevier B.V.},
	Title = {Some simplified {NP}-complete graph problems},
	Volume = {1},
	Year = {1976}}

@InProceedings{guibas1991,
author="Guibas, Leonidas J.
and Hershberger, John E.
and Mitchell, Joseph S. B.
and Snoeyink, Jack Scott",
editor="Hsu, Wen-Lian
and Lee, R. C. T.",
title="Approximating polygons and subdivisions with minimum link paths",
booktitle="ISA'91 Algorithms",
year="1991",
publisher="Springer Berlin Heidelberg",
address="Berlin, Heidelberg",
pages="151--162",
isbn="978-3-540-46600-0"
}

@article{Hadlock1975,
	Author = {Hadlock, F.},
	Issn = {0097-5397},
	Journal = {SIAM Journal on Computing},
	Language = {eng},
	Number = {3},
	Pages = {221--225},
	Publisher = {Society for Industrial and Applied Mathematics},
	Title = {Finding a Maximum Cut of a Planar Graph in Polynomial Time},
	Volume = {4},
	Year = {1975}}

@article{Harper1966,
	Author = {Harper, L.H.},
	Issn = {0021-9800},
	Journal = {Journal of Combinatorial Theory},
	Language = {eng},
	Number = {3},
	Pages = {385--393},
	Title = {Optimal numberings and isoperimetric problems on graphs},
	Volume = {1},
	Year = {1966}}

@article{hellcruz17,
	Author = {P. Hell and C. Hern\'andez-Cruz},
	Journal = {Discrete Applied Mathematics},
	Pages = {609-617},
	Title = {Strict Chordal and Strict Split Digraphs},
	Volume = {216},
	Year = {2017}}

@article{Jinjiang1995,
	Author = {Jinjiang, Yuan and Sanming, Zhou},
	Day = {01},
	Doi = {10.1007/BF02662875},
	Issn = {1993-0445},
	Journal = {Applied Mathematics},
	Month = {Sep},
	Number = {3},
	Pages = {337--344},
	Title = {Optimal labelling of unit interval graphs},
	Url = {https://doi.org/10.1007/BF02662875},
	Volume = {10},
	Year = {1995}}

@article{LaMarM.Drew2010Sd,
year = {2012},
title = {Split digraphs},
author = {LaMar, M. Drew},
journal = {Discrete Mathematics},
volume = {312},
number = 7,
pages = {1314 - 1325},
}

@article{MaňuchJán2008FPCT,
issn = {0219-7200},
journal = {Journal of Bioinformatics and Computational Biology},
pages = {93--106},
volume = {6},
publisher = {Imperial College Press},
number = {1},
year = {2008},
title = {FITTING PROTEIN CHAINS TO CUBIC LATTICE IS {NP}-COMPLETE},
language = {eng},
author = {Ma\v{n}uch, J\'{a}n and Gaur, Daya Ram},
}

@article{orlovdorf72,
    Author = {G I Orlova and Y G Dorfman},
    Journal = {Engineering Cybernetics},
    publisher = {Scripta Pub Co},
    volume = {10},
    pages = {502 - 506},
    title = {Finding the Maximum Cut in a Graph},
    year = {1972},
    
}

@article{petit11,
	Author = {Jordi Petit},
	Journal = {Bulletin of the {EATCS}},
	Number = {105},
	Pages = {177 - 201},
	Title = {Addenda to the Survey of Layout Problems},
	Year = {2011}}

@phdthesis{safro02,
	Author = {I Safro},
	School = {Weizmann Institute of Science.},
	Title = {The minimum linear arrangement problem},
	Year = {2002}}

@article{shiloach79,
	Author = {Yossi Shiloach},
	Journal = {SIAM Journal on Computing},
	Number = {1},
	Pages = {15 - 32},
	Title = {A minimum linear arrangement algorithm for undirected trees},
	Volume = {8},
	Year = {1979}}

@MastersThesis{tippenhauer2016,
    author     =     {Simon Tippenhauer},
    title     =     {On Planar 3-{SAT} and its Variants},
    school     =     {Freien Universit\"at Berlin},
    address     =     {Berlin
Germany},
    year     =     {2016},
    }

@article{Bodlaender91,
	Author = {Bodlaender, H. L. and Jansen, Klaus},
	Journal = {Nordic Journal of Computing},
	Language = {eng},
	Title = {On the complexity of the Maximum Cut problem},
	Year = {1991}}

\end{document}